\documentclass[runningheads,a4paper,envcountsame]{llncs} 

\usepackage{graphicx}

\usepackage[cmex10]{amsmath}
\usepackage{amssymb}

\usepackage{url}
\urldef{\mailsa}\path|pradeep.banerjee@gmail.com|
  
\newcommand{\keywords}[1]{\par\addvspace\baselineskip
\noindent\keywordname\enspace\ignorespaces#1}

\begin{document}

\mainmatter  % start of an individual contribution

\title{A Secret Common Information Duality \\for Tripartite Noisy Correlations}
\titlerunning{Secret Common Information Duality}

\author{Pradeep Kr. Banerjee}
\authorrunning{P.K. Banerjee}

\institute{Indian Institute of Technology Kharagpur\\
\mailsa}

\toctitle{Lecture Notes in Computer Science}
\tocauthor{Authors' Instructions}
\maketitle

\begin{abstract}
We explore the duality between the simulation and extraction of secret correlations in light of a similar well-known operational duality between the two notions of common information due to Wyner, and G{\'a}cs and K\"{o}rner. For the \emph{inverse} problem of simulating a tripartite noisy correlation from noiseless secret key and unlimited public communication, we show that Winter's (2005) result for the key cost in terms of a conditional version of Wyner's common information can be simply reexpressed in terms of the existence of a bipartite protocol monotone. For the \emph{forward} problem of key distillation from noisy correlations, we construct simple distributions for which the conditional G{\'a}cs and K\"{o}rner common information achieves a tight bound on the secret key rate. We conjecture that this holds in general for non-communicative key agreement models. We also comment on the interconvertibility of secret correlations under local operations and public communication.
\end{abstract}

\section{Introduction}
Information-theoretic (IT) or unconditional security---widely acknowledged as the strictest notion of security has witnessed a renaissance since the 90's, arguably following Maurer's seminal work on secret key (SK) agreement by public discussion from correlated source sequences \cite{ref1}, \cite{ref2}. Compared to computational complexity-based approaches, IT-security provides a framework for \emph{provable security} by bounding the adversary's total information. Assumptions are made neither on the latter's computational or memory resources nor on the unproven computational hardness of certain problems. By harnessing noise and appropriate coding and signaling strategies at the \emph{physical layer} (the lowest layer in the protocol stack), IT-security potentially complements conventional upper layer cryptographic protocols (e.g., RSA, Diffie-Hellman key exchange) and is an important component of future communication networks \cite{ref19}. The paradigm is especially valuable for improving security in large-scale wireless networks and distributed networks with minimal infrastructure (e.g., mobile \emph{ad hoc} sensor networks) where the broadcast nature of the transmission medium makes it particularly vulnerable to attacks, and key distribution and management is difficult and computationally expensive \cite{ref19}. 

But for all its advantages, unconditionally secure key distribution is impossible to realize from scratch, for instance, if the legitimate parties are only given access to a noiseless public communication channel \cite{ref1,ref2,ref3}. This pessimism can however be relativized if the parties are additionally given access to simple auxiliary devices (e.g., noisy correlations or communication channels) that are not completely under the control of an adversary \cite{ref14,ref15}. Thus, information-theoretic reductions between such primitives are of great interest. 

Interconvertibility of noisy tripartite correlations and a uniformly distributed, noiseless SK have been studied under the rubric of SK agreement by public discussion (called the \emph{forward problem} \cite{ref1,ref2}) and the dual problem of simulating a noisy correlation from noiseless SK and public communication (called the \emph{inverse problem} \cite{ref3,ref4,ref5,ref6}). In general, there are irreversible losses in exchanging noisy correlations in that, going from one noisy correlation to another and back is not lossless (even asymptotically) \cite{ref3,ref17}. In light of the resource character of noisy correlations in enabling SK agreement, and that of SK in simulating a tripartite correlation, quantifying such resources are of interest.

To make things precise, consider three distant parties, honest Alice and Bob, and an adversary Eve who observe sequences ${X^n}=(X_1,\ldots,X_n)$, ${Y^n}=(Y_1,\ldots,Y_n)$, and ${Z^n}=(Z_1,\ldots,Z_n)$ respectively, where the sequence triple $(X^nY^nZ^n)$ has the generic component variables $(XYZ)\sim p_{XYZ}$. Starting with no initially shared SKs, and using only local operations and unlimited public communication (LOPC) over a noiseless, authenticated (but, otherwise insecure) channel, what is the maximum rate at which Alice and Bob can distil a SK (i.e., the maximum possible \emph{SK rate}), such that Eve's information ($Z^n$ and the entire public discussion) about the generated SK is arbitrarily small? Conversely, starting with perfect SKs, what is the \emph{SK cost} of approximately simulating the correlated triple $XYZ$ using only LOPC?

We explore this duality between the secrecy extractable from $p_{XYZ}$ and that required to simulate $p_{XYZ}$ in light of a similar well-known operational duality between the two notions of common information (CI) due to G{\'a}cs and K\"{o}rner \cite{ref7} and Wyner \cite{ref8}. For the \emph{inverse} problem of simulating $p_{XYZ}$ from noiseless SK and unlimited public communication, Winter \cite{ref4} gave a single-letter characterization of the asymptotic minimal SK cost of formation in terms of a conditional version of Wyner's CI. We first show that the SK cost of formation can be simply reexpressed in terms of the existence of a bipartite protocol monotone. For the \emph{forward} problem of key distillation from $p_{XYZ}$, we construct simple distributions for which the conditional G{\'a}cs and K\"{o}rner CI captures the ``explicit'' secret CI, thus achieving a tight bound on the SK rate. We also comment on the interconvertibility of secret correlations under LOPC.

\thinmuskip=0mu
\section{CI Duality and Secret CI}
Random variables (RVs) and their finite alphabets are denoted using uppercase letters $X$ and script letters $\mathcal{X}$. ${X^n}$ denotes the sequence $(X_1,\ldots,X_n)$. $p_X(x)=\Pr\{X = x\}$ denotes the distribution (pmf) of a discrete RV $X$. $X-Y-Z$ denotes that $X,Y,Z$ form a Markov chain satisfying ${p_{XYZ}} = {p_{XY}}{p_{Z|Y}}$. Likewise, $X,Y,Z$ is said to form a conditional Markov chain given $U$ if $X-UY-Z$.
The entropy of $X$ is defined as $H(X)= -\sum\nolimits_{x \in \mathcal{X}} {p_X(x)\log p_X(x)}$ and the mutual information of $X$ and $Y$ is given by $I(X;Y)=H(X)-H(X|Y)$. The total variational distance between $p_X$ and $p_{X'}$ is defined as $\mathsf{TV}(p_X, p_{X'}) = \tfrac{1}{2}\sum\nolimits_{x \in \mathcal{X}}|p_X(x)-p_{X'}(x)|$.

The zero pattern of ${p_{XY}}$ can be specified by its characteristic bipartite graph $B_{XY}$ with the vertex set $\mathcal{X} \cup \mathcal{Y}$ and an edge connecting two vertices $x$ and $y$ iff ${p_{XY}}(x,y) \geq 0$. If $B_{XY}$ contains only a single connected component, we say that ${p_{XY}}$ is \emph{indecomposable}. An \emph{ergodic decomposition} of $p_{XY}(x,y)$ is defined by a unique partition of the space $\mathcal{X}\times\mathcal{Y}$ into connected components. The following double markovity lemma \cite{ref12} (also see Problem 16.25, p. 392 in \cite{ref11}) is useful. 
\begin{lemma}
A triple of RVs $(X,Y,Q)$ satisfies the double Markov conditions 
\begin{align*}
X-Y-Q,\text{ }Y-X-Q \tag{1}
\end{align*}
iff there exists a pmf $p_{Q'|XY}$ such that $H(Q'|X)=H(Q'|Y)=0$ and $XY-Q'-Q$. Furthermore, (1) implies $I(XY;Q)\leq H(Q')$ with equality iff $H(Q'|Q)=0$.
\end{lemma}
\begin{proof}
Given $p_{Q|XY}$ such that $X-Y-Q$ and $Y-X-Q$, it follows that $p_{XY}(x,y)>0 \Rightarrow p_{Q|XY}(q|x,y)=p_{Q|X}(q|x)=p_{Q|Y}(q|y)$ $\forall q$. Given an ergodic decomposition of $p_{XY}(x,y)$ such that $\mathcal{X}\times\mathcal{Y}=\bigcup\nolimits_{{q'}} \mathcal{X}_{q'}\times\mathcal{Y}_{q'}$, where the $\mathcal{X}_{q'}$$'$s and $\mathcal{Y}_{q'}$$'$s having different subscripts are disjoint, define $p_{Q'|XY}$ as $Q'=q'$ iff $x \in \mathcal{X}_{q'} \Leftrightarrow y \in \mathcal{Y}_{q'}$. Clearly $H(Q'|X)=H(Q'|Y)=0$. Then, for any $Q=q$ and for every $q'$, $p_{Q|XY}(q|\cdot,\cdot)$ is constant over $\mathcal{X}_{q'}\times\mathcal{Y}_{q'}$. This implies that $p_{Q|XY}(q|x,y)=p_{Q|Q'}(q|q')$ so that $XY-Q'-Q$. The converse is obvious. Thus, given (1), we get $Q'$ such that $I(XY;Q|Q')=0$ so that $I(XY;Q)=I(XYQ';Q)=I(Q';Q)=H(Q')-H(Q'|Q)\leq H(Q')$.
\qed \end{proof}

G{\'a}cs and K\"{o}rner (GK) \cite{ref7} defined CI as the maximum rate of common randomness (CR) $(R)$ that Alice and Bob, observing $X^n$ and $Y^n$ separately can \emph{extract} without communication $(R_0=0)$, i.e., $C_{GK}(X;Y)=\sup \tfrac{1}{n}H(f_1(X^n))$ where the sup is taken over all sequences of pairs of deterministic mappings $(f_1^n,f_2^n)$ such that $\Pr \{ f_1^n({X^n}) \ne f_2^n({Y^n})\}  \to 0{\text{ as }}n \to \infty$ (see setup in Fig. \ref{fig:fig_sim}(a)). GK showed that 
\begin{align*}
{C_{GK}}(X;Y) =\mathop {\max }\limits_{\substack{Q:\text{ }H(Q|X) = H(Q|Y) = 0}} H(Q)= H(Q_*) \tag{2}
\end{align*}
where $Q_*$ is the \emph{maximal common RV} of the pair $(X,Y)$ induced by the ergodic decomposition of ${p_{XY}}$. For all $X,Y$, we have $I(X;Y)=H(Q_*)+I(X;Y|Q_*)$. We say that $p_{XY}$ is \emph{resolvable}, if $I(X;Y|Q_*)=0$. An alternative characterization of $C_{GK}(X;Y)$ follows from Lemma 1 \cite{ref12}.
\begin{align*}
{C_{GK}}(X;Y) = \mathop {\max }\limits_{\substack{ Q:\text{ }Q-X-Y, \text{ } Q-Y-X}} I(XY;Q),{\text{ }}|\mathcal{Q}| \leq |\mathcal{X}||\mathcal{Y}|+2 \tag{3}
\end{align*}
$C_{GK}(X;Y)$ is identically zero for all indecomposable distributions. For example, a binary symmetric channel with non-zero crossover probability is indecomposable and hence $C_{GK}(X;Y)=0$. Thus, CR is a far stronger resource than correlation, in that the latter does not result in common random bits, in general. Nevertheless, when Alice communicates with Bob ($R_0>0$), they can unlock hidden layers of potential CR \cite{ref9}. With a high enough rate of communication (that is independent of Bob's output), the CR rate \emph{increases} to $I(X;Y)$ \cite{ref9}.

A conditional version of GK CI is defined as follows.
\begin{align*}
{C_{GK}}(X;Y|Z)= H(Q_*|Z)=\mathop {\max }\limits_{\substack{Q:H(Q|XZ)=0\\ \hspace{3mm} H(Q|YZ) = 0}} H(Q|Z) = \mathop {\max }\limits_{\substack{ Q-XZ-Y \\ Q-YZ-X}} I(XY;Q|Z) \tag{4}
\end{align*}
Conditioning always reduces GK CI, i.e., ${C_{GK}}(X;Y|Z)$ $\leq$ ${C_{GK}}(X;Y)$. We say that $p_{XYZ}$ is \emph{conditionally resolvable}, if $I(X;Y|ZQ_*)=0$. 

Wyner \cite{ref8} defined CI as the minimum rate of CR $(R)$ needed to \emph{generate} $X^n$ and $Y^n$ separately using local operations (independent noisy channels: $Q \to X^n$, $\;$ $Q \to Y^n$) and no communication $(R_0=0)$ (see setup in Fig. \ref{fig:fig_sim}(b)).
\begin{align*}
{C_W}(X;Y) = \mathop {\min }\limits_{Q:X - Q - Y} I(XY;Q),{\text{ }}|\mathcal{Q}| \leq |\mathcal{X}||\mathcal{Y}| \tag{5}
\end{align*}
Likewise, a conditional version of Wyner's CI is defined as follows.
\begin{align*}
{C_W}(X;Y|Z) &= \mathop {\min }\limits_{Q:X - QZ - Y} I(XY;Q|Z) \tag{6} 
\end{align*}
${C_W}(X;Y)$ quantifies the resource cost for the distributed approximate simulation of ${p_{XY}}$. When Alice communicates with Bob ($R_0>0$), with a high enough rate of communication (independent of Bob's output), the CR rate \emph{reduces} to $I(X;Y)$ \cite{ref9}. Reversing the direction of Alice's operation $(X^n \to Q^n)$ (see Fig. \ref{fig:fig_sim}(c)) leads to a two-stage simulation of a noisy channel via the Markov chain $X - Q - Y$. Now Alice and Bob can use the reverse Shannon theorem \cite{ref10} to simulate a first stage, with Alice mapping $X^n$ to some intermediate RV $Q^n$ which she sends noiselessly to Bob. In the second stage, Bob locally maps $Q^n$ to get $Y^n$. This gives a nontrivial tradeoff between the (noiseless) communication rate $R_0$ and CR rate $R$ leading to an alternative characterization of Wyner's CI as the communication cost of distributed channel simulation without any CR $(R = 0)$. With unlimited CR, the cost reduces to $I(X;Y)$ \cite{ref9}. Finally, to complete the duality, we note the following well-known relation between the different notions of CI \cite{ref12}: ${C_{GK}}(X;Y) \leq I(X;Y) \leq {C_W}(X;Y)$ with equality holding iff $p_{XY}$ is resolvable, whence ${C_{GK}}(X;Y) = I(X;Y) \Leftrightarrow I(X;Y) = {C_W}(X;Y)$.
\begin{figure}[t]
\centering
\includegraphics[width=4.8in,height=3.2in,keepaspectratio]{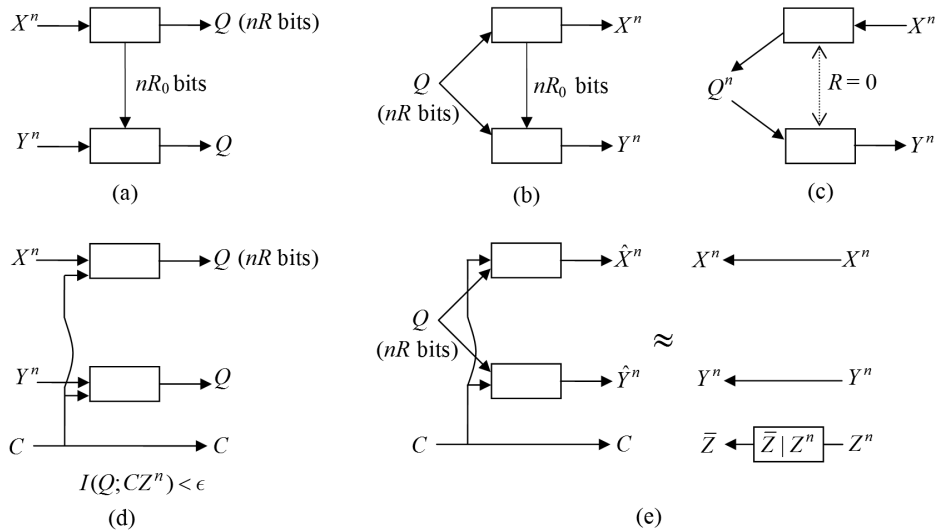}
\caption{\emph{Common information duality (with communication)}: (a) G{\'a}cs and K\"{o}rner's setup $(R_0=0)$ \cite{ref7} (b) Wyner's setup $(R_0=0)$ \cite{ref8} (c) Inverting Alice's channel gives an alternative characterization of Wyner's CI as the communication cost of channel simulation without any CR $(R = 0)$ \cite{ref9,ref10} (d) Setup for SK agreement \cite{ref1,ref2,ref3} (e) Setup for distributed simulation of $p_{XYZ}$ from SK and public communication \cite{ref3,ref4,ref5,ref6}}
\label{fig:fig_sim}
\end{figure}

The setup for the standard SK agreement scenario \cite{ref1,ref2,ref3,ref14,ref15,ref16} shown in Fig. \ref{fig:fig_sim}(d) is a generalization of the GK setup (see Fig. \ref{fig:fig_sim}(a)) that now allows for interactive communication. Consider the following distributed communication protocol, $\Pi_{KA}$. Alice ($X$) and Bob ($Y$) communicate interactively over an authenticated (noiseless) public discussion channel (transparent to an adversary, Eve ($Z$)). Both have independent access to an infinite stream of private randomness. The protocol proceeds in rounds, where in each round each party flips private coins, and based on the messages exchanged so far, publicly sends a message to the other party. At the end of the protocol, Alice (Bob) either accepts or rejects the protocol execution, and outputs a key $Q_X$ ($Q_Y$) depending on her (his) view of the protocol, which comprises of $X^n$ ($Y^n$), all local computations, and the entire public communication (encapsulated in the RV $C$). The asymptotic maximum rate of SK distillation is called the \emph{SK rate} \cite{ref3,ref14}. 

\spnewtheorem{defn}{Definition}{\bfseries}{\itshape}
\begin{defn}
The SK rate, $S(X;Y||Z)$ is defined as the largest real number $R$ such that for all $\epsilon > 0$, there exists an integer $n$ and a randomized protocol $\Pi_{KA}$ with communication $C$ that with probability $1-\epsilon$, allows Alice (knowing $X^n$) and Bob (knowing $Y^n$) to compute $Q_X$ and $Q_Y$, respectively, satisfying $\Pr \{{{Q}_{X}}={{Q}_{Y}}=Q\}\ge 1-\epsilon ,\text{ }H(Q)=\log |\mathcal{Q}|\ge n(R-\epsilon),\text{ }I(Q;CZ^n)<\epsilon$. 
\end{defn}

Analogously, there exists an inverse protocol $\Pi_{form}$, and a dual measure, the asymptotic minimal \emph{SK cost of formation} of the triple $(X,Y,Z) \sim p_{XYZ}$ \cite{ref3,ref4,ref5,ref6}. While the distributed channel synthesis problem \cite{ref9} (setup in Fig. \ref{fig:fig_sim}(b)) explores the role of CR and one-way communication in the formation of bipartite correlations ${p_{XY}}$ shared between honest Alice and Bob, the setup in Fig. \ref{fig:fig_sim}(e) explores the role of \emph{secret CR} (i.e., SK) and public communication in the distributed simulation of tripartite correlations ${p_{XYZ}}$, now additionally shared with an adversary Eve \cite{ref3,ref4,ref5,ref6}. To start with, Alice and Bob share a SK in the form of $R$ perfectly correlated bits, i.e., Alice has $Q_X$ and Bob, $Q_Y$ such that $Q_X=Q_Y=Q$ and $H(Q)=R$. The goal is to approximately simulate correlated sequence triples of the form $(X^n,Y^n,Z^n)$, upto a local degrading of $Z^n$. Both parties have independent access to sources of private randomness. The protocol proceeds in rounds and at the end of the communication phase, Alice and Bob output ${\hat {X}^n}$ and ${\hat {Y}^n}$, respectively, as deterministic functions of their views, which comprises of the shared SK $Q$, all the private coins flipped so far, and the entire public communication ($C$). Discounting the private randomness by allowing for stochastic mappings, $\Pi_{form}$ can be specified by the joint distribution $p_{{\hat {X}^n}{\hat {Y}^n}CQ}$ $=$ $p_{{{\hat {X}^n}|CQ},{{\hat {Y}^n}|CQ}}p_{CQ}$. The resulting simulation is ``good'' from Alice and Bob's point of view, if they end up correctly generating the marginal $p_{X^nY^n}$. However, Eve ($Z^n$) might still have some information about $X^nY^n$. This is formalized by requiring that Eve's optimal channel $p_{{\bar Z}|Z^n}$ is only able to simulate the public communication ($C$) used by Alice and Bob to generate $({{\hat {X}^n},{\hat {Y}^n}})$. The following definition makes this precise \cite{ref3}.
\begin{defn}
The SK cost of formation, $SK_c(X;Y|Z)$ is the infimum of all numbers $R \geq 0$ such that for all $\epsilon > 0$, there exists an integer $n$ and a randomized protocol $\Pi_{form}$  with communication $C$ that with probability $1-\epsilon$, allows Alice and Bob, knowing a common random $\left\lfloor {nR} \right\rfloor$-bit string $Q$, to compute ${{\hat X}^n}$ and ${{\hat Y}^n}$, respectively, such that $\Pr \{ ({{\hat X}^n},{{\hat Y}^n},C) = ({X^n},{Y^n},\bar Z)\}  \geq 1 -\epsilon$ holds for some correlated sequence triple $(X^n,Y^n,Z^n)$ that has generic component variables $(X,Y,Z)\sim p_{XYZ}$ and some channel $p_{\bar Z|Z^n}$.
\end{defn}

${p_{XYZ}}$ contains secret correlations iff it cannot be generated by LOPC, i.e., $SK_c(X;Y|Z)>0$ \cite{ref18}. Both the SK cost of formation and the SK rate are measures of the secrecy content of $p_{XYZ}$ and admit a clear operational interpretation: $SK_c(X;Y|Z)$ quantifies the \emph{minimum} amount of SK bits required to (approximately) \emph{simulate} $p_{XYZ}$, whereas $S(X;Y||Z)$ quantifies the \emph{maximum} amount of SK bits that one can \emph{extract} from $p_{XYZ}$. 

Winter used resolvability-based arguments \cite{ref4} to arrive at the full secret correlation vs. public communication trade-off for the inverse problem and defined the SK cost of formation with unlimited public communication as
\begin{align*}
{SK_c}(X;Y|Z) = \mathop {\min }\limits_{\substack{
  Q:X - Q\bar Z - Y  \\
  \bar Z:XY - Z - \bar Z  \\ 
}}  I(XY;Q|\bar Z) \tag{7}
\end{align*}
where the minimum is taken over all RVs $Q$ and $\bar Z$. Cardinalities of the corresponding alphabets are bounded as $|{\mathcal{Q}}| \leq |{\mathcal{X}}||{\mathcal{Y}}|$, and $|\bar{\mathcal{Z}} | \leq  |\mathcal{Z}|$ respectively. ${SK_c}(X;Y|Z)$ is bounded from below by the \emph{intrinsic information} \cite{ref3}, that intuitively speaking, measures the correlation shared between Alice and Bob that Eve cannot access or destroy \cite{ref2}.
\begin{align*}
I(X;Y \downarrow Z) = \mathop {\min }\limits_{\bar Z:XY - Z - \bar Z} I(X;Y|\bar Z)  \tag{8}
\end{align*}
where the cardinality of the alphabet $\bar {\mathcal{Z}}$ of RV $\bar Z$ is bounded as $|\bar {\mathcal{Z}}| \leq |\mathcal{Z}|$ \cite{ref13}. 
Both $I(X;Y \downarrow Z)$ and ${SK_c}(X;Y|Z)$ are known to be \emph{lockable} \cite{ref3,ref4}, i.e., can fall sharply by an arbitrary large amount on giving away a single bit to Eve.

\section{Main contributions}
\subsection{Wyner's Conditional CI and the SK Cost of Formation}
${C_W}(X;Y)$ can be interpreted as the SK cost of creating a product distribution with Eve. Now consider Eve's optimal channel $p_{\bar{Z}|Z}$. To create a (public) correlation, Alice randomly samples $\bar{Z}$ and publicly announces the value of $\bar{Z}$ over a symmetric broadcast channel to Bob and Eve. Then the SK cost of creating the product distribution $p_{XY|{\bar{Z}}}$ is $C_W(X;Y|{\bar{Z}})$. In this section, we introduce the concept of protocol monotones \cite{ref6,ref14,ref15,ref16,ref17} for axiomatizing the general properties of upper bounds on the SK rate and show that $C_W(X;Y|\bar{Z})$ qualifies as such an upper bound and is a natural candidate for quantifying the SK cost of formation.
 
When distant parties wish to establish a SK by manipulating some set of private and public resources, it is natural to restrict attention to LOPC operations. Since the ``resourcefulness'' or ``secrecy content'' of the state is a non-local property that cannot increase under LOPC, this set of transformations is deemed as a \emph{free} resource. Mathematically, resources can be quantified by \emph{monotones}, real-valued functions of joint distributions that cannot increase under LOPC. LOPC monotones were first introduced in \cite{ref17} as classical counterparts of entanglement or LOCC (local operations and classical communication) monotones to study the rate of resource conversion under LOPC. A resource cannot increase under LOPC operations by Alice and Bob. Since Eve, in her role as a malicious adversary is always assumed to operate optimally against Alice and Bob, the same resource cannot decrease under her LOPC operations \cite{ref6,ref14,ref17}. We define a monotone as follows.
\begin{defn}
For all jointly distributed RVs $(X,Y,Z)$, let $\mathcal{M}(X;Y|Z)$ be a real-valued function of $p_{XYZ}$. Then $\mathcal{M}$ is a monotone if the following hold:

\begin{trivlist}
  \item \emph{1) Monotonicity under local operations (LO) by Alice and Bob:} Suppose Alice modifies $X$ to ${\bar X}$ by sending $X$ over a channel, characterized by $p_{{\bar X}|{X}}$. Then $\mathcal{M}$ can only decrease, i.e., for all jointly distributed RVs $(X,Y,Z,\bar X)$ with $\bar X:YZ - X - \bar X$, $\mathcal{M}(\bar X;Y|Z)$ $\leq$ $\mathcal{M}(X;Y|Z)$, and likewise for Bob.
  \item \emph{2) Monotonicity under public communication (PC) by Alice and Bob:} Suppose Alice publicly announces the value of $\widetilde X$. Then $\mathcal{M}$ can only decrease, i.e., for all jointly distributed RVs $(X,Y,Z, \widetilde X)$ with $H(\widetilde X|X)=0$, $\mathcal{M}(X;Y\widetilde X|Z\widetilde X)$ $\leq$ $\mathcal{M}(X;Y|Z)$, and likewise for Bob.
  \item \emph{3) Monotonicity under local operations (LO) by Eve:} Suppose Eve modifies $Z$ to $\bar Z$ by sending $Z$ over a channel, characterized by $p_{\bar Z|{Z}}$. Then $\mathcal{M}$ can only increase, i.e., for all jointly distributed RVs $(X,Y,Z,\bar Z)$ with $\bar Z:XY - Z - \bar Z$,  $\mathcal{M}(X;Y|\bar Z)$ $\geq$ $\mathcal{M}(X;Y|Z)$.
  \item \emph{4) Monotonicity under public communication (PC) by Eve:} Suppose Eve publicly announces the value of $\widetilde Z$. Then $\mathcal{M}$ can only increase, i.e., for all jointly distributed RVs $(X,Y,Z,\widetilde Z)$ with $H(\widetilde Z|Z)=0$,, $\mathcal{M}(\widetilde ZX;\widetilde ZY|Z)$ $\geq$ $\mathcal{M}(X;Y|Z)$.
  \item \emph{5) Additivity and Continuity:} $\mathcal{M}$ is additive on tensor products and is a semi-positive, continuous function of $p_{XYZ}$. A stronger notion of asymptotic continuity requires that for two pmfs ${{p_{XYZ}}}$, ${{q_{XYZ}}}$, if $\mathsf{TV}({{p_{XYZ}}}, {{q_{XYZ}}})=\epsilon$, then $|\mathcal{M}({{p_{XYZ}}})-\mathcal{M}({{q_{XYZ}}})|\leq \epsilon \log{d} + \delta(\epsilon)$, where $d$ is some constant that depends on $|\mathcal{X}|,\;|\mathcal{Y}|$ and $|\mathcal{Z}|$ and $\delta(\epsilon)$ is any function that depends only on $\epsilon$ with $\delta(0)=0$.
\end{trivlist}
\end{defn}

If $\mathcal{M}$ satisfies the conditions in Definition 3, then $\mathcal{M}$ is an upper bound on the SK rate \cite{ref14,ref15}.
Upper bounds on the rate at which instances of a target primitive can be realized per invocation of the source primitive can be obtained by comparing the value of the monotone on the source and target states. Given the class of LOPC operations, suppose that the parties are able to convert $n$ copies of $p_{XYZ}$ into some realization of a distribution $q'_{XYZ}$ which is close to $m$ independent realizations of the target distribution $q_{XYZ}$, i.e., ${p^{ \otimes n}}\mathop  \to \limits^{LOPC} q' \simeq {q^{ \otimes m}}$. Then by virtue of Property (5) in Definition 3, we have, $\mathcal{M}({p^{ \otimes n}})=n\mathcal{M}(p)\geq \mathcal{M}(q')\simeq m\mathcal{M}(q)$, so that the optimal rate  $R(p \to q)$ at which transformations ${p^{ \otimes n}}\mathop  \to \limits^{LOPC} {q^{ \otimes m}}$ are possible is upper bounded as  $\tfrac{m}{n} \leq \tfrac{\mathcal{M}({p})}{\mathcal{M}({q})}$ \cite{ref17}. Thus if $p_{SS{\Delta}}^s$ denotes the distribution of a perfect secret bit with $p_{SS{\Delta}}^s(0,0,\delta)=p_{SS{\Delta}}^s(1,1,\delta)=\tfrac{1}{2}$, then for any pmf $q_{XYZ}$, $S(X;Y||Z)=R(q \to p^s)\leq \tfrac{I(X;Y \downarrow Z)}{I(S;S \downarrow \Delta)} = I(X;Y \downarrow Z)$, since $I(S;S \downarrow \Delta)=1$, and $I(X;Y \downarrow Z)$ is a monotone \cite{ref3}. Likewise $R(p^s \to q)\leq \tfrac{I(S;S \downarrow \Delta)}{I(X;Y \downarrow Z)}=\tfrac{1}{I(X;Y \downarrow Z)}$. Hence $SK_c(X;Y|Z)=\tfrac{1}{R(p^s \to q)} \geq I(X;Y \downarrow Z)$ \cite{ref3}.

$C_W(X;Y|Z)$ violates monotonicity under LO by Eve\footnote{Given jointly distributed RVs $(X,Y,Z,\bar Z)$ with $\bar Z:XY - Z - \bar Z$, it does not hold in general that $C_W(X;Y|\bar{Z}) \geq C_W(X;Y|Z)$.} and thus fails to achieve an upper bound on the SK rate. To preserve monotonicity under LO by Eve, an additional minimization is required over all stochastic maps ${p_{\bar Z|Z}}$ that can be chosen by Eve for locally processing her observations. This implies a minimization over all Markov chains $XY - Z - \bar Z$ which is tantamount to the simulation approximation of $(X^nY^nZ^n)$ upto a local degrading in $Z^n$ in Definition 2 of the SK cost of formation. This yields a new quantity called \emph{Wyner's intrinsic conditional CI} that satisfies monotonicity under Eve's LOPC operations. 
\begin{align*}
  {C_W}(X;Y \downarrow Z) = \mathop {\min }\limits_{\bar Z:XY - Z - \bar Z} {C_W}(X;Y|\bar Z)
    = \mathop {\min }\limits_{\substack{
  Q:X - Q\bar Z - Y \\
  \bar Z:XY - Z - \bar Z  \\ 
}}  I(XY;Q|\bar Z)  \tag{9}
\end{align*}
The cardinality of the alphabet $\bar{\mathcal{Z}}$ of $\bar Z$ is bounded as $|\bar{\mathcal{Z}}| \leq |\mathcal{Z}|$. This can be shown following similar arguments as in \cite{ref13} using Carath{\'e}odory's theorem. From (6), we have $|\mathcal{Q}| \leq |\mathcal{X}||\mathcal{Y}|$. 
We now show that ${C_W}(X;Y \downarrow Z)$ is indeed a monotone and hence a valid upper bound on the SK rate.

\begin{lemma}
${C_W}(X;Y \downarrow Z)$ is a monotone.
\end{lemma}

\begin{proof}
Let $p_{\bar{Z}|Z}$ be the minimizer in ${C_W}(X;Y \downarrow Z)$.  Then it suffices to show that ${C_W}(X;Y|\bar{Z})$ satisfies monotonicity under Alice and Bob's LOPC operations. We shall require the following monotonicity inequalities: (a) $H(Z|f(Y)) \geq H(Z|Y)$, (b) $I(X;Y|f(X)Z) \leq I(X;Y|Z)$. For all jointly distributed RVs $(XY\bar{Z}Q\bar{X})$ such that $Y\bar{Z} - X - \bar X$ and $\bar{X}-Q\bar{Z}-Y$, monotonicity under LO by Alice holds since, 
$$  I({\bar X}Y;Q|\bar{Z})=H(Q|\bar{Z})-H(Q|{\bar X}Y\bar{Z})\mathop\leq\limits^{{\text{(a)}}}H(Q|\bar{Z})-H(Q|XY\bar{Z})=I(XY;Q|\bar{Z}),$$
and likewise for Bob. Similarly, for all jointly distributed RVs $(XY\bar{Z}Q\widetilde X)$ such that $H(\widetilde X|X)=0$, monotonicity under PC by Alice holds since,
$$  I(XY\widetilde{X};Q\widetilde{X}|\bar{Z}\widetilde{X}) =  I(XY;Q|\bar{Z}\widetilde{X})\mathop \leq  \limits^{{\text{(b)}}} I(XY;Q|\bar{Z}),$$
and likewise for Bob. 

Additivity and continuity are valuable properties if the monotone is to provide information on the rate of transformations ${p^{ \otimes n}}\mathop  \to \limits^{LOPC} {q^{ \otimes m}}$ in the asymptotic limit $n,m \to \infty$. It is easy to show that for independent triples $(X_1Y_1\bar{Z_1})$ and $(X_2Y_2\bar{Z_2})$, $C_W(X_1X_2;Y_1Y_2|{\bar{Z_1}\bar{Z_2}})=C_W(X_1;Y_1|\bar{Z_1})+C_W(X_2;Y_2|\bar{Z_2})$. We do not target the stronger notion of asymptotic continuity since it is already known that the SK cost fails this property and admits locking \cite{ref4}.
\qed
\end{proof}
Theorem 3 gives the main result of this section.
\begin{theorem}
With unlimited public communication, ${SK_c}(X;Y|Z)={C_W}(X;Y\downarrow{Z})$. Furthermore, if ${p_{\bar Z|Z}}$ is the optimal stochastic map achieving the minimum in the characterization of ${C_W}(X;Y \downarrow Z)$, then ${C_W}(X;Y \downarrow Z)=I(X;Y \downarrow Z)$ iff there exists a pmf $p_{Q'|XYZ\bar{Z}}$ s.t. $H(Q'|X\bar{Z})=H(Q'|Y\bar{Z})=0$ and $X-\bar{Z}Q'-Y$. 
\end{theorem}

\begin{proof}
From (7), (9) and Lemma 2, the first equality is obvious. For the second equality, the ``only if'' part is trivial. For the ``if'' part, first note that
\begin{align*}
  {C_W}(X;Y \downarrow Z) &= \mathop {\min }\limits_{\substack{  Q:X - Q\bar Z - Y \\  \bar Z:XY - Z - \bar Z}}  I(XY;Q|\bar Z)= \mathop {\min }\limits_{\substack{Q:X - Q\bar Z - Y \\
  \bar Z:XY - Z - \bar Z}}  (I(Y;Q|X\bar Z) + I(X;Q|\bar Z)) \\ 
  &\mathop =  \limits^{{\text{(a)}}} \text{ }\mathop {\min }\limits_{\substack{  Q:X - Q\bar Z - Y \\  \bar Z:XY - Z - \bar Z}}  (I(X;Y|\bar{Z})+I(Y;Q|X\bar Z) + I(X;Q|Y\bar Z)),
\end{align*}
where (a) follows from writing $I(X;Q|\bar Z)$ as $I(X;Y|\bar Z)+I(X;Q|Y\bar Z)-I(X;Y|Q\bar Z)$, and noting that $I(X;Y|Q\bar Z)=0$. 
Then given jointly distributed RVs $(XYZ\bar{Z}Q)$ achieving the minimum in (9), if $I(Y;Q|X\bar Z)=I(X;Q|Y\bar Z)=0$, then by Lemma 1, there exists a pmf $p_{Q'|XYZ\bar{Z}}$ such that $H(Q'|X\bar{Z})=H(Q'|Y\bar{Z})=0$, $XY-Q'\bar{Z}-Q$ and $I(XY;Q|\bar{Z})\leq H(Q'|\bar{Z})$, with equality iff $H(Q'|Q\bar{Z})=0$. Finally note that $H(Q'|\bar{Z})\leq I(X;Y|\bar{Z})$ with equality iff $X-\bar{Z}Q'-Y$.
\qed
\end{proof}

\subsection{G{\'a}cs and K\"{o}rner (GK) Conditional CI and the SK Rate}
For the general source model (with arbitrary public discussion), calculation of the exact SK rate remains an open problem \cite{ref1,ref2,ref3,ref16}. Without any communication, the problem is still more complicated in that even the determination of probability distributions that allow for the distillation of SK remains an unsolved problem. GK showed that in a non-communicative model, common codes of a pair of discrete, memoryless correlated sources cannot exploit any correlation beyond a certain deterministic interdependence of the sources. Remarkably, they showed that the asymptotic case is no better than the zero error case and that $C_{GK}(X;Y)$ depends only on the zero pattern of the underlying joint distribution ${p_{XY}}$. Intuitively, ${C_{GK}}(X;Y|Z)$ captures the most ``explicit'' form of secret CI, i.e., the maximum amount of SK that can be extracted without any communication. In this section, we construct simple distributions for which ${C_{GK}}(X;Y|Z)$ achieves a tight bound on the SK rate. We conjecture that ${C_{GK}}(X;Y|Z)$ quantifies the achievable SK rate for non-communicative key agreement models. We also comment on the lossless interconvertibility of secret correlations.
\spnewtheorem{exmpl}{Example}{\bfseries}{\itshape}
\begin{exmpl}
\emph{Consider the distribution $p_{XYZ}^1$ with $\mathcal{X}=\mathcal{Y}=\mathcal{Z}=\{0,1,2,3\}$. We write $p_{XYZ}^1(a,b,c)=(abc)$. Given $(000)=(011)=(101)=(110)=\tfrac{1}{8}$, and $(222)=(333)=\tfrac{1}{4}$. Graphically, $p_{XYZ}^1$ is shown in the following table (with $Z$'s value given in parentheses): $p_{XYZ}^1=\tfrac{1}{8}\left(\begin{smallmatrix}1(0)&1(1)&.&.\\1(1)&1(0)&.&.\\.&.&2(2)&.\\.&.&.&2(3)\end{smallmatrix}\right)$. Alice, Bob and Eve each have access to this table and many independent copies of $X$, $Y$ and $Z$, respectively. For each independent random experiment generating $(X,Y,Z)\sim p_{XYZ}^1$, Eve can infer Alice and Bob's values with complete certainty, when she receives either 2 or 3. When she receives either 0 or 1, she can only infer that Alice's and Bob's symbols are restricted to the upper left quadrant (i.e., the set $\{0,1\}$), but then in this range, $X$ and $Y$ are uniformly distributed. Clearly, Alice and Bob can share no secret. 
Now, consider the distributions $p_{XYZ}^2=\tfrac{1}{8}\left(\begin{smallmatrix}1(0)&1(1)&.&.\\1(1)&1(0)&.&.\\.&.&2(2)&.\\.&.&.&2(2)\end{smallmatrix}\right)$, and  $p_{XYZ}^3=\tfrac{1}{8}\left(\begin{smallmatrix}1(0)&1(1)&.&.\\1(1)&1(0)&.&.\\.&.&2(0)&.\\.&.&.&2(1)\end{smallmatrix}\right)$, where Eve's symbol set is successively depleted to $\mathcal{Z}=\{0,1,2\}$ and $\mathcal{Z}=\{0,1\}$, respectively.
In the latter case, Alice and Bob can realize one perfect SK bit, since Eve can no longer infer anything about their quadrant information (upper left or lower right). 
This intuition is borne out by ${C_{GK}}(X;Y|Z)$, which evaluates to 0, 0.5 and 1, respectively, for $p^1$, $p^2$ and $p^3$.
$p^3$ is \emph{resolvable} but not \emph{conditionally resolvable}. A distribution $p^4$ that is \emph{conditionally resolvable} but not \emph{resolvable} is the following: $(000)=(100)=(010)=(110)=\tfrac{1}{18},\text{ }(001)=(101)=(011)=(111)=\tfrac{1}{21},\text{ } (032)=(042)=\tfrac{1}{15},\text{ } (252)=\tfrac{1}{5},\text{ } (220)=\tfrac{1}{9},\text{ } (221)=\tfrac{1}{7}$. 
Finally, $C_{GK}(X;Y|Z)=1$ for the following more general distribution $p^5$ that is neither \emph{resolvable} nor \emph{conditionally resolvable}: $(000)=(011)=(020)=(101)=\tfrac{1}{10},\text{ }(110)=(121)=\tfrac{1}{20},\text{ }(230)=(331)=\tfrac{1}{4}$, where using arguments similar to the ones for $p^3$, Alice and Bob can be shown to achieve one perfect SK bit.}
\end{exmpl}

\begin{exmpl}
\emph{Consider the following sequence of distributions:
\begin{align*}
  {({p_{XYZ}})_n}(x,y,z)&= \begin{cases}
  \tfrac{1}{2n^2}, {\text{ }}\operatorname{if}{\text{ }}x,y \in \{0,\ldots,n-1\},{\text{ }} z=x+y{\text{ }}(\bmod n), \\ 
  \tfrac{1}{2n}, {\text{ }}\operatorname{if}{\text{ }}x \in \{n,\ldots,2n-1\},{\text{ }} y=x,{\text{ }}z=x{\text{ }}(\bmod n).
  \end{cases}\\
  {({q_{XYZ}})_n}(x,y,z)&= \begin{cases}
  \tfrac{1}{{\log n}}{({p_{XYZ}})_n},{\text{ }}\operatorname{if}{\text{ }} x \ne \Delta,{\text{ }}y \ne \Delta,{\text{ }}z \ne \Delta, \\
  1 - \tfrac{1}{{\log n}},{\text{ }}\operatorname{if}{\text{ }} x = y = z = \Delta,
  \end{cases}
\end{align*}
where ${({q_{XYZ}})_n}$ is derived from ${({p_{XYZ}})_n}$ by extending the symbol sets $({\mathcal{X}},{\mathcal{Y}},{\mathcal{Z}})$ to include an extra symbol $\Delta$. Renner and Wolf \cite{ref3} showed that ${({q_{XYZ}})_n}$ has \emph{asymptotic bound information}, i.e., ${({q_{XYZ}})_n}$ is asymptotically non-distillable $(S({X_{(n)}};{Y_{(n)}}||{Z_{(n)}}) = \tfrac{1}{{\log n}} \to 0$, as $n \to \infty)$, and yet cannot be created by LOPC, since $SK_c({X_{(n)}};{Y_{(n)}}|{Z_{(n)}}) \geq I({X_{(n)}};{Y_{(n)}} \downarrow {Z_{(n)}}) = \tfrac{1}{{\log n}}(1 + \tfrac{1}{2}\log n) > \tfrac{1}{2}$ \cite{ref3}. Omitting the index $n$, for a given $Z$, the maximal common RV of the pair $(X,Y)$ has the pmf ${p_{Q_*}}(0) = \tfrac{1}{{2\log n}},{\text{ }}{p_{Q_*}}(1) = 1 - \tfrac{1}{{\log n}},{\text{ }}{p_{Q_*}}(i) = \tfrac{1}{{2n\log n}},{\text{ }}i = 2,...,n + 1$, so that $C_{GK}(X;Y|Z)=H(Q_*|Z)=\tfrac{1}{{\log n}}$, thus achieving the SK rate.}
\end{exmpl}
In both the examples above, no assumptions have been made on the nature of the communication protocol (or its lack thereof) between Alice and Bob. For a non-communicative SK agreement model, i.e. a restricted key agreement model where no communication is allowed between Alice and Bob, let $S_{0-\text{comm}}(X;Y||Z)$ denote the maximum attainable rate at which a SK can be extracted by Alice (knowing $X^n$) and Bob (knowing $Y^n$) about which Eve (knowing $Z^n$) has virtually no information. As usual, $(X^n,Y^n,Z^n)$ are independent repeated realizations of the random experiment $p_{XYZ}$. 
\spnewtheorem{conjec}[theorem]{Conjecture}{\bfseries}{\itshape}
\begin{conjec}
For non-communicative SK agreement models, $S_{0-\text{comm}}(X;Y||Z)$ $={C_{GK}}(X;Y|Z)$.
Furthermore, if $p_{XYZ}$ is conditionally resolvable, then $$S_{0-comm}(X;Y||Z)={C_{GK}}(X;Y|Z)=I(X;Y|Z).$$
\end{conjec}

\section{Concluding Remarks}
In summary, we have presented two results. The first result demonstrates the power of the framework of resource monotones for axiomatizing the general properties of upper bounds on the SK rate. We showed that $C_W(X;Y|Z)$ violates monotonicity under LO by Eve, which can be used to bootstrap the definition of the SK cost of simulating the triple $(XYZ)$ up to a local degrading of $Z$. In the past, a similar approach has been put to direct practical use for deriving strong bounds on the SK rate \cite{ref16}. 
Given many independent copies of a source distribution ($q_{XYZ}$) and a target distribution of a perfect secret bit ($p_{SS{\Delta}}^s$), \emph{reversible} conversion of $q_{XYZ}$ and $p_{SS{\Delta}}^s$ is possible iff $R(q \to p^s)R(p^s \to q)=1$, i.e., iff $S(X;Y||Z)=I(X;Y \downarrow Z)$ and $I(X;Y \downarrow Z)=SK_c(X;Y|Z)$. In Theorem 3, we have given necessary and sufficient conditions for achieving the second equality. 
More generally, monotones can be used to derive upper bounds on the asymptotic rate of conversion of $q_{XYZ}$ to any other distribution $q_{X'Y'Z'}$ (e.g., one which is relatively less favorable to Eve). Thus, resource theories hold promise in capturing a general calculi of secret correlations.

In the second part, we constructed simple distributions for which the conditional G{\'a}cs and K\"{o}rner CI achieves a tight bound on the SK rate. The sequence of distributions ${({q_{XYZ}})_n}$ in Example 2 was originally constructed by Renner and Wolf \cite{ref3} to show that there exists sequences of distributions with asymptotic bound information, i.e., asymptotically non-distillable correlations with positive SK cost. This is the best that has been shown so far for the bipartite case. For more than two parties, the possibility of creating different bipartitions across the honest parties immensely simplifies the problem and, indeed, multipartite bound information has been shown to exist \cite{ref18}. Constructing distributions that can show the existence of bipartite bound information remains an open problem and a scope for future work. Another problem of independent interest is to consider single shot versions of the SK cost of formation and intrinsic information.

\section*{Acknowledgments}
PKB wishes to thank Amin Gohari for valuable comments and Paul Cuff for short useful discussions over email. This work is based in part on PKB's master's thesis and was supported in part by the AVLSI Consortium, IIT Kharagpur.

\end{document}